\documentclass[onefignum,onetabnum]{siamart190516}

\usepackage{lipsum}
\usepackage{amsfonts}
\usepackage{graphicx}
\usepackage{epstopdf}
\usepackage{algorithmic}
\ifpdf
  \DeclareGraphicsExtensions{.eps,.pdf,.png,.jpg}
\else
  \DeclareGraphicsExtensions{.eps}
\fi

\usepackage{siunitx} 
\usepackage{booktabs}

\newsiamremark{remark}{Remark}
\newsiamremark{hypothesis}{Hypothesis}
\crefname{hypothesis}{Hypothesis}{Hypotheses}
\newsiamthm{claim}{Claim}

\headers{Resolvability of Hamming Graphs - PREPRINT}{L. Laird, R. C. Tillquist, S. Becker, and M. E. Lladser}

\title{Resolvability of Hamming Graphs - PREPRINT\thanks{
\funding{This work has been partially funded by the NSF grant No. 1836914.}}}

\author{Lucas Laird\thanks{Department of Applied Mathematics, University of Colorado, Boulder.}
\and Richard C. Tillquist\thanks{Department of Computer Science, University of Colorado, Boulder.}
\and Stephen Becker\footnotemark[2]
\and Manuel E. Lladser\footnotemark[2] \thanks{Corresponding author.}
\email{(manuel.lladser@colorado.edu)}
}

\usepackage{amsopn}

\ifpdf
\hypersetup{
  pdftitle={Resolvability of Hamming Graphs (PREPRINT)},
  pdfauthor={L. Laird, R.C. Tillquist, S. Becker, and M.E. lladser}
}
\fi

\newcommand{\kmer}{\hbox{$k$-mer}}
\newcommand{\kmers}{\hbox{$k$-mers}}

\newcommand{\Trace}{{\hbox{Tr}}}

\newcommand{\vct}{{\hbox{vec}}}

\newcommand{\V}{\mathbb{V}}
\newcommand{\Hka}{\mathbb{H}_{k,a}}

\usepackage{stmaryrd}

\begin{document}

\maketitle

\begin{abstract}
A subset of vertices in a graph is called resolving when the geodesic distances to those vertices uniquely distinguish every vertex in the graph. Here, we characterize the resolvability of Hamming graphs in terms of a constrained linear system and deduce a novel but straightforward characterization of resolvability for hypercubes. We propose an integer linear programming method to assess resolvability rapidly, and provide a more costly but definite method based on Gr\"obner bases to determine whether or not a set of vertices resolves an arbitrary Hamming graph. As proof of concept, we identify a resolving set of size 77 in the metric space of all octapeptides (i.e., proteins composed of eight amino acids) with respect to the Hamming distance; in particular, any octamer may be readily represented as a 77-dimensional real-vector. Representing k-mers as low-dimensional numerical vectors may enable new applications of machine learning algorithms to symbolic sequences.
\end{abstract}



\section{Introduction}

In what follows, $k\ge1$ and $a\ge2$ are fixed integers, and we refer to elements in the set $\V:=\{0,\ldots,a-1\}^k$ as $\kmers$, which we represent either as strings or row vectors depending on the context.

The Hamming distance between two $\kmers$ $u$ and $v$, from now on denoted as $d(u,v)$, is the number of coordinates where the $\kmers$ differ, and is a valid metric. The Hamming graph $\Hka$ has $\V$ as its vertex set, and two $\kmers$ $u$ and $v$ are adjacent (i.e. connected by an undirected edge) if and only if $d(u,v)=1$, i.e. $u$ and $v$ differ at exactly one coordinate. As a result, the (geodesic) distance between two vertices in $\Hka$ is precisely their Hamming distance (see Figure~\ref{fig:HamResEx}). The literature refers to the Hamming graph with $a=2$ as the ($k$-dimensional) hypercube. 

A non-empty set $R\subseteq\V$ is called resolving when for all $u,v\in\V$, with $u\ne v$, there exists $r\in R$ such that $d(u,r)\ne d(v,r)$. In other words, $R$ multilaterates $\V$. For instance, $\V$ resolves $\Hka$ because $d(u,v)=0$ if and only if $u=v$. Equivalently, $R\subseteq\V$ is resolving if and only if the transformation $\Phi:\V\to\mathbb{R}^{|R|}$ defined as $\Phi(v):=(d(v,r))_{r\in R}$ is one-to-one. In particular, the smaller a resolving set of $\Hka$, the lower the dimension needed to represent $\kmers$ as points in a Euclidean space, which may be handy e.g. to represent symbolic data numerically for machine learning tasks~\cite{TilLla19}.

\begin{figure}[h!]
\centering
\includegraphics[width = 0.33\textwidth]{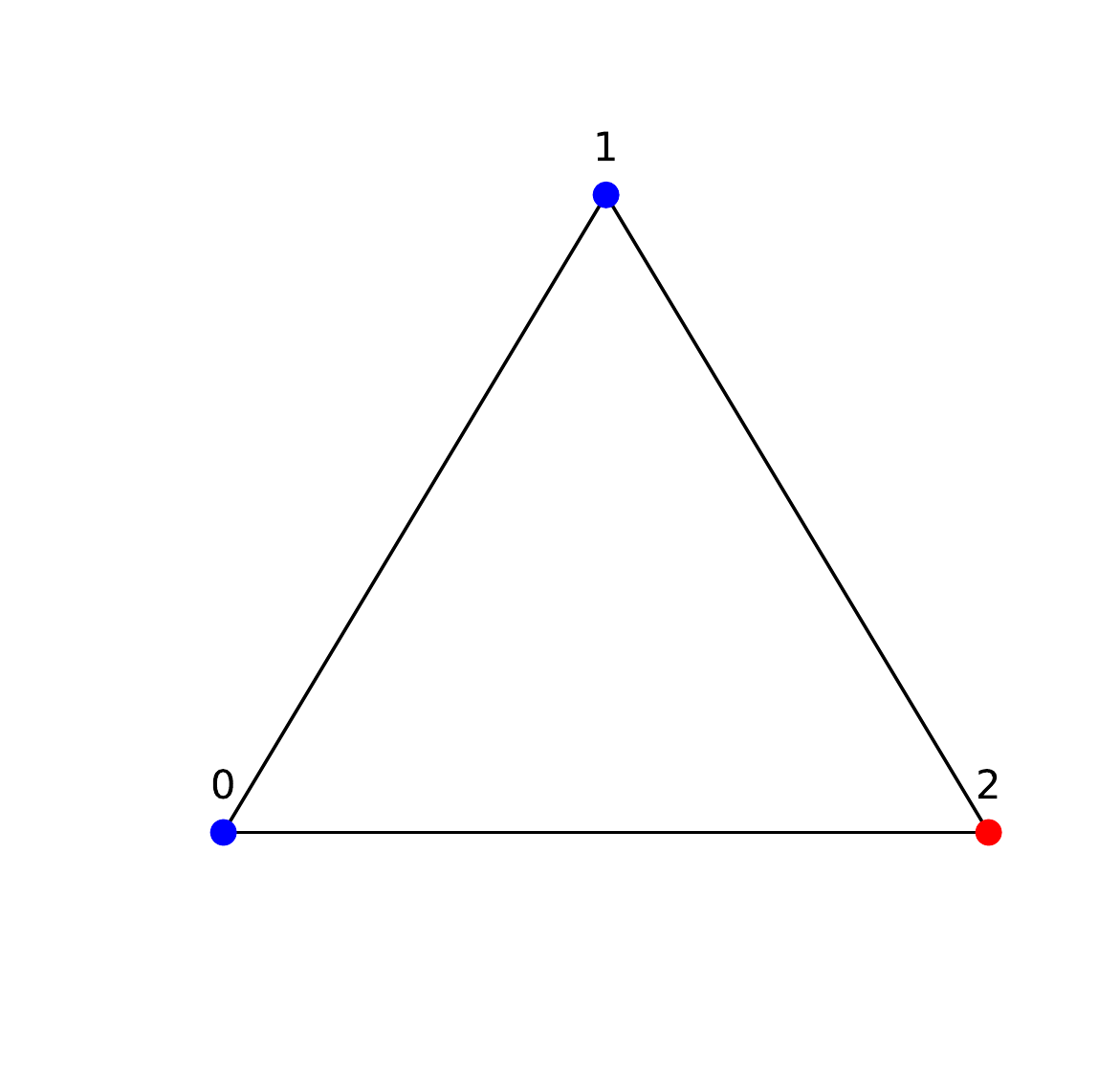}\includegraphics[width = 0.33\textwidth]{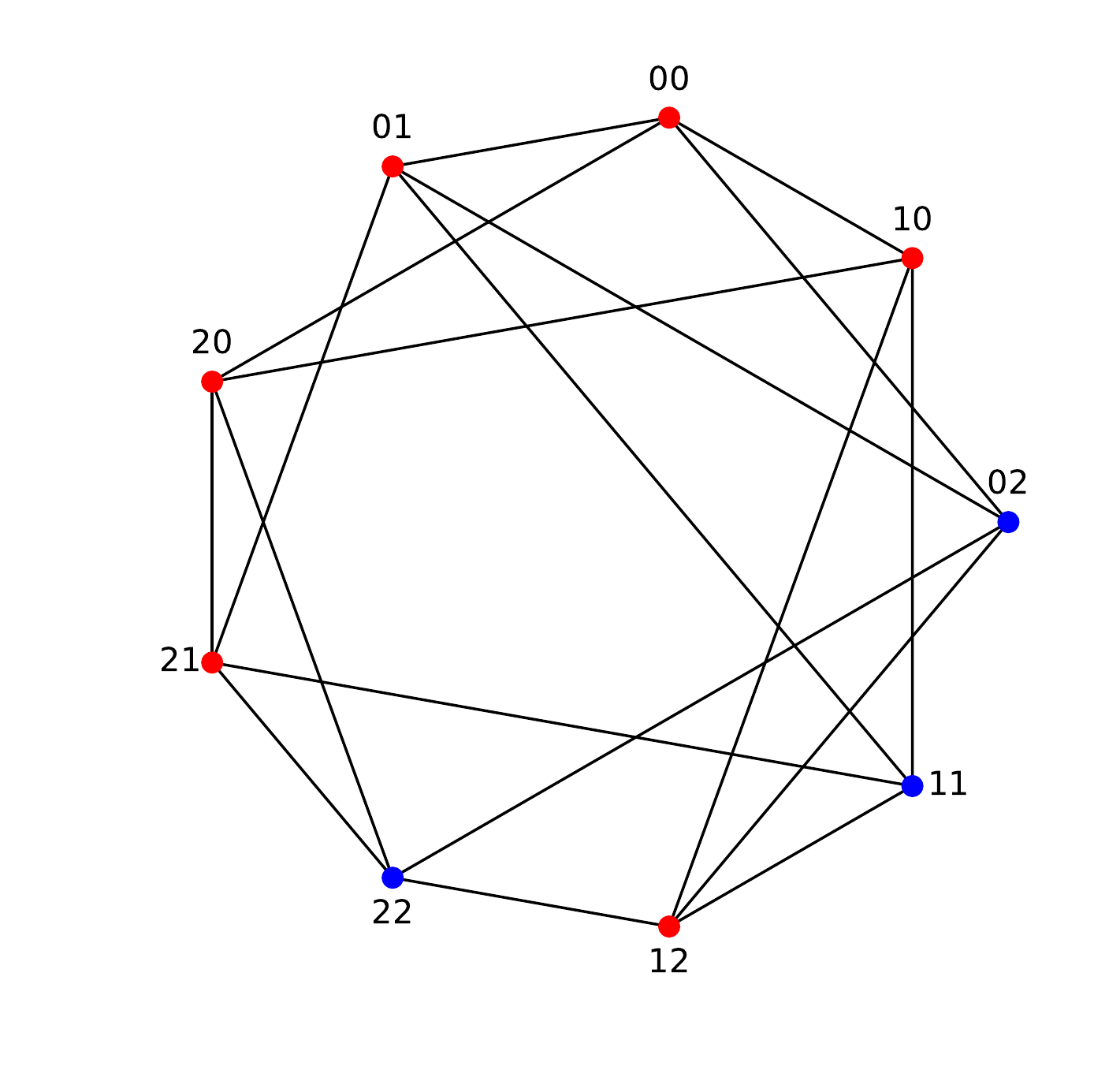}\includegraphics[width = 0.33\textwidth]{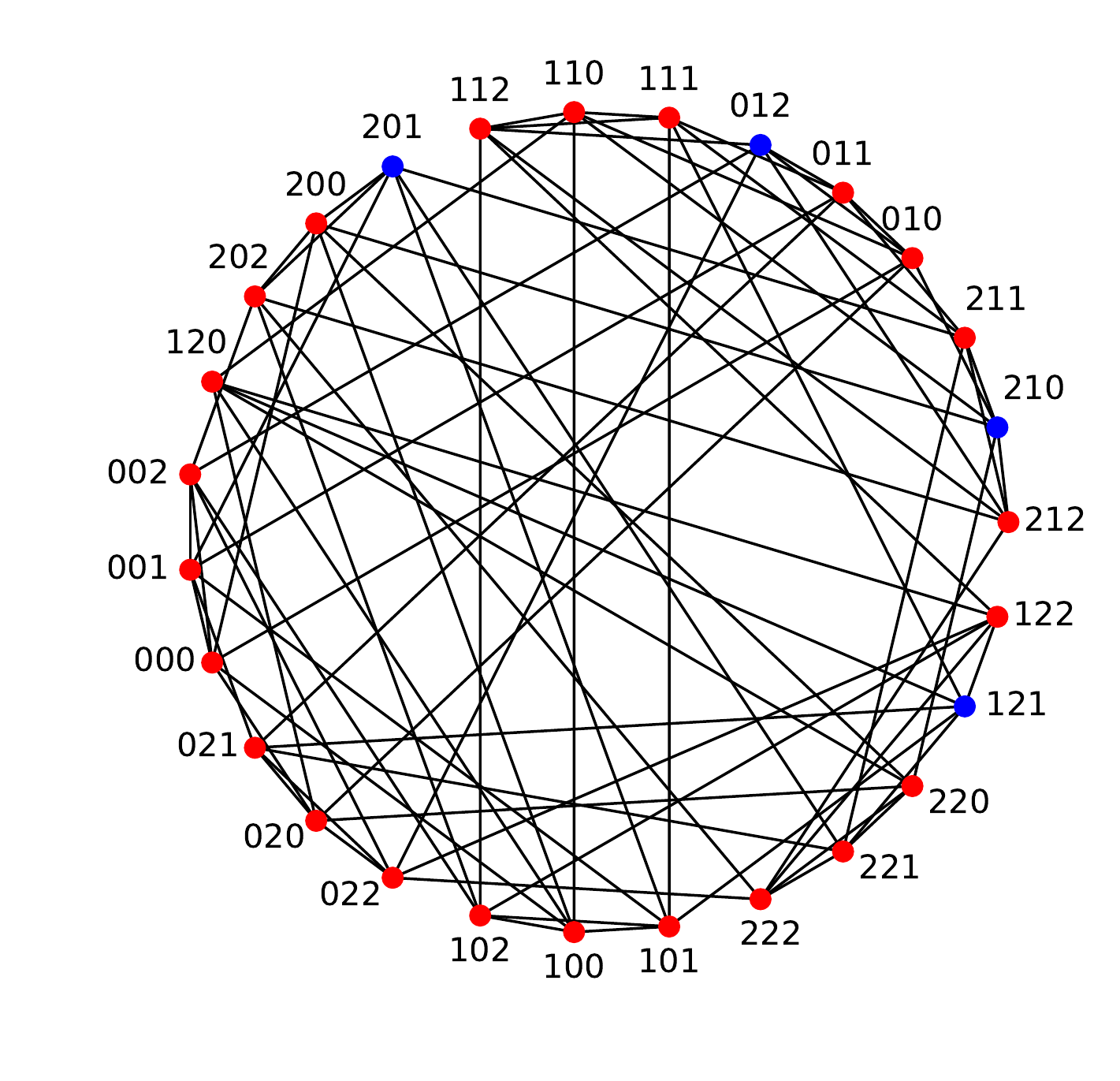}
\caption{Visual representation of $\mathbb{H}_{1,3}$, $\mathbb{H}_{2,3}$, and $\mathbb{H}_{3,3}$. Blue-colored vertices form minimal resolving sets in their corresponding Hamming graph.}
\label{fig:HamResEx}
\end{figure}

The metric dimension of $\Hka$, which we denote $\beta(\Hka)$, is defined as the size of a minimal resolving set in this graph~\cite{HarMel76,Sla75}. For instance, $\beta(\mathbb{H}_{1,a})=(a-1)$ because $\mathbb{H}_{1,a}$ is isomorphic to $K_{a-1}$, the complete graph on $(a-1)$ vertices~\cite{chartrand2000resolvability}. Unfortunately, computing the metric dimension of an arbitrary graph is a well-known NP-complete problem~\cite{Coo71,GarJoh79,KhuRagRos96}, and it remains unknown if this complexity persists when restricted to Hamming graphs. In fact, the metric dimension of hypercubes is only known up to dimension $k=10$~\cite{HarMel76}, and values have been conjectured only up to dimension $k=17$~\cite{MlaKraKovEtAl12}---see OEIS sequence A303735 for further details~\cite{OEIS19}.

\newpage

Integer linear programming (ILP) formulations have been used to search for minimal resolving sets~\cite{chartrand2000resolvability,currie2001metric}. In the context of Hamming graphs, a potential resolving set $R$ is encoded by a binary vector $y$ of dimension $a^k$ such that $y_j=1$ if $j\in R$ and $y_j=0$ if $j\in\V\setminus R$. One can then search for a minimal resolving set for $\Hka$ by solving the ILP~\cite{chartrand2000resolvability}:
\begin{alignat}{2} 
&\min\limits_y \;  &  &\sum_{j\in\V} y_j  \notag \\
&\text{subject to} \;  &  &\sum_{j\in\V} |d(u,j)-d(v,j)| \cdot y_j \ge 1, \; \forall u\ne v\in\V \\
 & &  & y_j \in \{0,1\}, \; \forall j\in\V. \notag
\label{eq:oldILP}
\end{alignat}
The first constraint ensures that for all pairs of different vertices $u$ and $v$, there is some $j\in R$ such that $|d(u,j)-d(v,j)| > 0$, hence $R$ resolves $\Hka$. The objective penalizes the size of the resolving set. (A variant due to~\cite{currie2001metric} is similar but stores $a^k$ copies of a binary version of the distance matrix of the graph.) One downside of this formulation is that forming the distance matrix of $\Hka$ requires $\mathcal{O}(a^{2k})$ storage, as well as significant computation. Moreover, standard approaches to reduce the computation below $\mathcal{O}(a^{2k})$, such as fast multipole methods~\cite{greengard1987fast} and kd-trees~\cite{bentley1975multidimensional}, do not obviously apply. Even if one could compute all pairwise distances between nodes, simply storing the distance matrix is impractical. To fix ideas, the graph $\mathbb{H}_{8,20}$---which is associated with octapeptides (see \cref{sec:protein_representation})---has $20^8$ nodes, so storing the distance matrix with $\log_2(8)=3$ bits per entry and taking advantage of symmetry would require $3{20^8\choose 2}$ bits, or approximately a prohibitive 123 exabytes.

Due to the above difficulties, other efforts have focused on finding small resolving sets rather than minimal ones. When $a^k$ is small, resolving sets for $\Hka$ may be determined using the so-called Information Content Heuristic (ICH) algorithm~\cite{HauSchVie12}, or a variable neighborhood search algorithm~\cite{MlaKraKovEtAl12}. Both approaches quickly become intractable with increasing $k$. However, the highly symmetric nature of Hamming graphs can be taken advantage of to overcome this problem. Indeed, recent work~\cite{TilLla19} has shown that $\beta(\Hka)\le\beta(\mathbb{H}_{k-1,a})+\lfloor a/2\rfloor$; in particular, $\beta(\Hka)\le(k-1)\lfloor a/2\rfloor+(a-1)$ i.e., just $\mathcal{O}(k)$ nodes are enough to resolve all the $a^k$ nodes in $\Hka$. Moreover, one can find a resolving set of size $\mathcal{O}(k)$ in only $\mathcal{O}(ak^2)$ time~\cite{TilLla19}. 

This manuscript is based on the recent Bachelor's thesis~\cite{Lai19}, and has two overarching goals. First, it aims to develop practical methods for certifying the resolvability, or lack thereof, of subsets of nodes in arbitrary Hamming graphs. So far, this has been addressed for hypercubes in the literature~\cite{Beardon:2013} but remains unexamined for arbitrary values of the parameter $a$. While our work does not directly address the problem of searching for minimal resolving sets, verifying resolvability is a key component of any such search and may shed new light on the precise metric dimension of $\Hka$ in future investigations. Second, this paper aims also to exploit said characterization to remove unnecessary nodes---if any---in known resolving sets. This problem, which is infeasible by brute force when $a^k$ is large, has not received any attention in the literature despite being crucial for the embedding of $\kmers$ into the Euclidean space of a lowest possible dimension.

The paper is organized as follows. Our main theoretical results are presented first in~\cref{sec:main}. \cref{thm:Az=0} provides the foundation from which we address the problem of verifying resolvability in Hamming graphs and implies a new characterization of resolvability of hypercubes (\cref{cor:symp_system}). An illustrative example shows the utility of \cref{thm:Az=0} but raises several practical challenges in its implementation on large Hamming graphs. \Cref{sec:grobner} describes a computationally demanding verification method based on Gr\"obner bases that is nevertheless more efficient than the brute force approach and determines with certainty whether or not a set of nodes in $\Hka$ is resolving. Computational issues are addressed in~\cref{sec:ILP} with a novel ILP formulation of the problem. This approach is fast but stochastic and hence has the potential to produce false positives or false negatives. \Cref{sec:complexity_experiments} compares the run time of these methods against a brute force approach across small Hamming graphs. Combining the techniques from sections~\ref{sec:grobner} and~\ref{sec:ILP}, \cref{sec:protein_representation} presents a simple approach to discovering and removing redundant nodes in a given resolving set. This approach allows us to improve on previous bounds on the metric dimension of the Hamming graph $\mathbb{H}_{8,20}$. Finally, two appendices provide background information about Gr\"obner bases and linear programming.

All code used in this manuscript is available on GitHub (\url{https://github.com/hamming-graph-resolvability/Hamming_Resolvability}).

\section{Main results}
\label{sec:main}

In what follows $\Trace(A)$ denotes the trace of a square matrix $A$, $B'$ the transpose of a matrix or vector $B$, and $\vct(C)$ the column-major ordering of a matrix $C$ i.e. the row vector obtained by appending from left to right the entries in each column of $C$. For instance:
\[\vct\left(\left[\begin{array}{cc} a & b \\ c & d\end{array}\right]\right)=(a,c,b,d).\]
In addition, $\bar D$ denotes the flip of the entries in a binary matrix (or vector) $D$, that is 0 is mapped to 1, and vice versa.

The one-hot encoding of a $\kmer$ $v$ is defined as the binary matrix $V$ of dimension $(a\times k)$ such that $V[i,j]=1$ if and only if $(i-1)=v[j]$ (the offset in $i$ is needed since the reference alphabet is $\{0,...,a-1\}$ instead of $\{1,\ldots,a\}$). Here, $V[i,j]$ denotes the entry in row-$i$ and column-$j$ of the matrix $V$, and similarly $v[j]$ denotes the $j$-th coordinate of the vector $v$. We also follow the convention of capitalizing $\kmer$ names to denote their one-hot encodings.

Our first result links one-hot encodings of $\kmer$s with their Hamming distance. Note this result applies to any alphabet size, not just binary.

\begin{lemma}\label{lem:UtV}
If $u,v$ are $\kmers$ with one-hot encodings $U,V$, respectively, then $d(u,v)=k-\Trace(U'V)$; in particular, $d(u,v)=\Trace(U'\bar V)$.
\end{lemma}
\begin{proof}
Let $U_i$ and $V_i$ be the $i$-th column of $U$ and $V$, respectively. Clearly, if $u[i]=v[i]$ then $\langle U_i,V_i\rangle = 1$, and if $u[i]\ne v[i]$ then $\langle U_i,V_i\rangle = 0$, because all but one of the entries in $U_i$ and $V_i$ vanish and the non-vanishing entries are equal to 1. As a result, $\Trace(U'V)=\sum_{i=1}^k\langle U_i,V_i\rangle$ counts the number of positions where $u$ and $v$ are equal; in particular, $d(u,v)=k-\Trace(U'V)$. Finally, observe that if $1^{a\times k}$ denotes the $(a\times k)$ matrix with all entries equal to 1 then $\Trace(U'1^{a\times k})=k$ because every row of $U'$ has exactly one 1 and all other entries vanish. As a result, $d(u,v)=\Trace(U'(1^{a\times k}-V))=\Trace(U'\bar V)$, as claimed.
\end{proof}

We can now give a necessary and sufficient condition for a subset of nodes in an arbitrary Hamming graph to be resolving.

\begin{theorem}\label{thm:Az=0}
Let $v_1,\ldots,v_n$ be $n\ge1$ $\kmers$ and $V_1,\ldots,V_n$ their one-hot encodings, respectively, and define the $(n\times ak)$ matrix with rows 
\begin{equation}
A := \left(\begin{array}{c}
\vct(V_1) \\
\vdots\\
\vct(V_n)
\end{array}\right).
\label{def:A}
\end{equation}
Then $R:=\{v_1,\ldots,v_n\}$ resolves $\Hka$ if and only if $0$ is the only solution to the linear system $Az=0$, with $z$ a column vector of dimension $ak$, satisfying the following constraints: if $z$ is parsed into $k$ consecutive but non-overlapping subvectors of dimension $a$, namely $z=((z_1,\ldots, z_a), (z_{a+1},\ldots, z_{2a}), ... , (z_{(k-1)a+1},\ldots,z_{ka}))'$, then each subvector is the difference of two canonical vectors.
\end{theorem}
\begin{proof}
Before showing the theorem observe that, for any pair of matrices $A$ and $B$ of the same dimension, $\Trace(A'B)=\langle\vct(A),\vct(B)\rangle$, where $\langle\cdot,\cdot\rangle$ is the usual inner product of real vectors.

Consider $\kmers$ $x$ and $y$, and let $X$ and $Y$ be their one-hot encodings, respectively. Due to \cref{lem:UtV}, $d(v_i,x)=d(v_i,y)$ if and only if $\Trace(V_i'(X-Y))=0$ i.e. $\langle\vct(V_i),\vct(X-Y)\rangle=0$. As a result, the set $R$ does not resolve $\Hka$ if and only if there are $\kmers$ $x$ and $y$ such that $Az\neq0$, where $z:=\vct(X)-\vct(Y)$. Note however that each column $X$ and $Y$ equals a canonical vector in $\mathbb{R}^a$; in particular, if we parse $\vct(X)$ and $\vct(Y)$ into $k$ subvectors of dimension $a$ as follows: $\vct(X)=(x_1,\ldots,x_k)$ and $\vct(Y)=(y_1,\ldots,y_k)$, then $z=(x_1-y_1,\ldots,x_k-y_k)$ with $x_i'$ and $y_i'$ canonical vectors in $\mathbb{R}^a$. This shows the theorem.
\end{proof}

\subsection{Illustrative Example}\label{subsec:Illustrative} In $H_{2,3}$ consider the set of nodes $R_0=\{02,11\}$. From~\cref{thm:Az=0}, $R_0$ resolves $H_{2,3}$ if and only if $A_0z=0$, with
\begin{equation}
A_0 = \begin{bmatrix}
    1 & 0 & 0 & 0 & 0 & 1 \\
    0 & 1 & 0 & 0 & 1 & 0 
\end{bmatrix},
\label{def:A0}
\end{equation}
has no non-trivial solution $z$ which satisfies the other constraints in the theorem when writing $z = \big((z_1,z_2,z_3),(z_4,z_5,z_6)\big)'$. Something useful to note about this decomposition is that if a subvector of $z$ has two identical entries, then all the entries in that subvector must vanish.

Note that $A_0$ is already in its reduced row echelon form~\cite{Olver:2018}, and has two pivots: $z_1 = -z_6$ and $z_2=-z_5$. Seeking non-trivial solutions to the constrained linear system, we examine permissible values for $z_5$ and $z_6$:
\begin{itemize}
\item[(a)] If $z_5=-1$ then we must have $(z_4,z_6)\in\{(0,1),(1,0)\}$. Furthermore, if $z_6=1$ then $(z_1,z_2,z_3)=(-1,1,0)$, but if $z_6=0$ then $(z_1,z_2,z_3)=(0,1,-1)$. Consequently, $z=(-1,1,0,0,-1,1)$ and $z=(0,1,-1,1,-1,0)$ solve the constrained system.
\item[(b)] Similarly, we find that $z=(-1,0,1,-1,0,1)$ and $z=(1,0,-1,1,0,-1)$ solve the constrained system when we assume that $z_5=0$.
\item[(c)] Finally, $z=(1,-1,0,0,1,-1)$ and $z=(0,-1,1,-1,1,0)$ are also found to solve the constrained system when we impose that $z_5=1$.
\end{itemize}
Having found at least one non-trivial solution to the constrained linear system, we conclude that $R_0$ does not resolve $H_{2,3}$. (The found $z$'s are in fact the only non-trivial solutions.)

From the proof of \cref{thm:Az=0}, we can also determine pairs of vertices in $H_{2,3}$ which are not resolved by $R_0$. Indeed, using the non-trivial solutions found above we find that 
$12$ and $01$, $21$ and $10$, and $00$ and $22$ are the only pairs of nodes in $H_{2,3}$ which are unresolved by $R_0$. In particular, because the distances between the nodes in each pair and $22$ are different, $R_1:=R_0\cup\{22\}$ resolves $H_{3,2}$.

We can double-check this last assertion noticing that the reduced echelon form of the matrix $A_1$ associated with $R_1$ is
\begin{equation}
    \hbox{rref}(A_1) = \begin{bmatrix}
    1 & 0 & 0 & 0 & 0 & 1 \\
    0 & 1 & 0 & 0 & 1 & 0 \\
    0 & 0 & 1 & 0 & 0 & 1
\end{bmatrix}.
\label{def:A1}
\end{equation}
In particular, $z_1 = -z_6$, $z_2 = -z_5$, and $z_3 = -z_6$. The first and third identity imply that $z_1=z_2$, hence $(z_1,z_2,z_3)=(0,0,0)$. This together with the first and second identity now imply that $(z_4,z_5,z_6)=(0,0,0)$. So, as anticipated, $z=0$ is the only solution to the constrained linear system $A_1z=0$.

In general, if the reduced row echelon form of the matrix given by \cref{thm:Az=0} has $j$ free variables, then there could be up to $3^j$ possible solutions to the associated linear system, each of which would have to be checked for the additional constraints. This exhaustive search could be very time consuming if not impossible. Handling the linear system constraints more systematically and efficiently is the motivation for sections~\ref{sec:grobner} and~\ref{sec:ILP}.

\subsection{Specializations to Hypercubes} In~\cite{Beardon:2013} a necessary and sufficient condition for the resolvability of hypercubes is provided exploiting that $d(u,v)=\|u-v\|_2^2$ when $u$ and $v$ are binary $\kmers$. Next, we reproduce this result using our framework of one-hot encodings instead.

\begin{corollary} \cite[Theorem 2.2]{Beardon:2013}
Let $R=\{v_1,\ldots,v_n\}$ be a set of $n\ge1$ binary $\kmers$, and define the $(n\times k)$ matrix with rows
\[B := \left[\begin{array}{c} 
v_1-\bar{v_1} \\ 
\vdots \\ 
v_n-\bar{v_n} 
\end{array}\right].\]
Then, $R$ resolves $H_{k,2}$ if and only if $\hbox{ker}(B)\cap\{0,\pm1\}^k=\{0\}$.
\label{cor:mat_system}
\end{corollary}

\begin{proof}
Let
\[A= \left[\begin{array}{ccccc}
\vline & \vline &  & \vline & \vline \\
A_1 & A_2 & \ldots & A_{2k-1} & A_{2k} \\
\vline & \vline & & \vline & \vline
\end{array}\right]\]
be the $(n\times 2k)$ matrix with columns $A_1,\ldots,A_{2k}$ given by \cref{thm:Az=0} for $R$. It follows that $R$ resolves $H_{k,2}$ if and only if $Az=0$, with $z=((x_1,y_1),\ldots,(x_k,y_k))'\in\{0,\pm1\}^{2k}$ and $(x_i+y_i)=0$ for each $i=1,\ldots,k$, has only a trivial solution. Note however that $Az=By$, where
\begin{eqnarray*}
B &:=& 
\left[\begin{array}{ccc}
\vline & & \vline\\
(A_2-A_1) & \ldots & (A_{2k}-A_{2k-1}) \\
\vline & & \vline
\end{array}\right];\\
y &:=& (y_1,\ldots,y_k)'.
\end{eqnarray*}
Therefore $R$ is resolving if and only if $By=0$, with $y\in\{0,\pm1\}^k$, has only a trivial solution. But recall from~\cref{thm:Az=0} that the rows of $A$ are the column-major orderings of the one-hot encodings of the binary $k$-mers in $R$. In particular, using $\llbracket\cdot\rrbracket$ to denote Iverson brackets, we find that the row in $B$ associated with $v\in R$ is:
\[\Big(\llbracket v[1]=1\rrbracket-\llbracket v[1]=0\rrbracket,\ldots,\llbracket v[k]=1\rrbracket-\llbracket v[k]=0\rrbracket\Big)=v-\bar v,\]
from which the corollary follows.
\end{proof}

We can provide an even simpler characterization of sets of $\kmers$ that resolve the hypercube, provided that $1^k:=(1,\ldots,1)$ is one of them. This seemly major assumption is only superficial. Indeed, hypercubes are transitive; that is, given any two binary $\kmers$ there is an automorphism (i.e., a distance preserving bijection $\sigma:\{0,1\}^k\to\{0,1\}^k$) that maps one into the other~\cite[\S3.1]{TilLla19}.  Hence, given any set $R$ of binary $\kmers$ there is an automphism $\sigma$ such that $1^k\in\sigma(R)$. In particular, because $R$ is resolving if and only if $\sigma(R)$ is resolving, one can assume without any loss of generality that $1^k$ is an element of $R$.

\begin{corollary}
Let $R=\{v_1,\ldots,v_n\}$ be a set of $n$ binary $\kmers$ such that $1^k\in R$, and define the $(n\times k)$ matrix with rows
\[C := \left[\begin{array}{c} 
v_1 \\ 
\vdots \\ 
v_n
\end{array}\right].\]
Then, $R$ resolves $H_{k,2}$ if and only if $\hbox{ker}(C)\cap\{0,\pm1\}^k=\{0\}$.
\label{cor:symp_system}
\end{corollary}

\begin{proof}
Note that for all binary $\kmer$s $v$: $(v+\bar v)=1^k$; in particular, $(v-\bar v)=(2v-1^k)$. Hence, if $B$ is as given in~\cref{cor:symp_system} and $C$ as defined above then
\[Bz=0\hbox{ if and only if }Cz=\langle 1^k,z\rangle\left[\begin{array}{c}1/2\\\vdots\\1/2\end{array}\right].\]
But, because $1^k\in R$ and $2\cdot1^k-1^k=1^k$, one of the entries in $Bz$ equals $\langle 1^k,z\rangle$. Since the entries in $Cz$ are proportional to $\langle 1^k,z\rangle$, $Bz=0$ if and only if $Cz=0$, from which the corollary follows.
\end{proof}

\section{Polynomial Roots Formulation}
\label{sec:grobner}

In this section, we express the constraints of the linear system in \cref{thm:Az=0} as roots of a multi-variable polynomial system, and we reveal various properties about this system which can drastically reduce the complexity of determining whether a subset of nodes resolves or not $\Hka$.

In what follows, for any given non-empty set $P$ of polynomials in a possibly multi-variable $z$, $\{P=0\}$ denotes the set of $z$'s such that $p(z)=0$, for each $p\in P$. Unless otherwise stated, we assume that $z$ has dimension $ka$, i.e. $z=(z_1,\ldots,z_{ka})$.

Consider the polynomial sets
\begin{eqnarray}
P_1 &:=& \Big\{z_i^3-z_i,\hbox{ for $i=1,\ldots,ka$}\Big\};\notag \\
P_2 &:=& \left\{\sum_{j=(i-1)a+1}^{ia}z_j,\hbox{ for $i=1,\ldots,k$}\right\};\\
P_3 &:=& \left\{\Big(2-\sum_{j=(i-1)a+1}^{ia}z_j^2\Big)\cdot\sum_{j=(i-1)a+1}^{ia}z_j^2,\hbox{ for $i=1,\ldots,k$}\right\}. \notag
\label{eq:P}
\end{eqnarray}
Our first result characterizes the constraints of the linear system in \cref{thm:Az=0} in terms of the roots of the above polynomials. Ahead, unless otherwise stated:
\begin{equation}
P := (P_1\cup P_2\cup P_3).
\label{def:P}
\end{equation}

\begin{lemma}
$z\in\{P=0\}$ if and only if when parsing $z$ into $k$ consecutive but non-overlapping subvectors of dimension $a$, each subvector is the difference of two canonical vectors.
\label{lem:polsys}
\end{lemma}
\begin{proof}
The polynomials in $P_1$ enforce that each entry in $z$ must be a ${-1}$, $0$, or $1$, while the polynomials in $P_2$ enforce that there is a $(-1)$ for every $1$ in each subvector of $z$. Finally, the polynomials in $P_3$ enforce that each subvector of $z$ has exactly two non-zero entries or no non-zero entries. Altogether, $z\in P$ if and only if each subvector is identically zero, or it has exactly one 1 and one $(-1)$ entry and all other entries vanish, i.e. each subvector of $z$ is the difference of two canonical vectors in $\mathbb{R}^a$.
\end{proof}

The following is now an immediate consequence of this lemma and~\cref{thm:Az=0}.

\begin{corollary}
Let $R$ be a set of nodes in $\Hka$ and $A$ the matrix given by equation~\cref{def:A}. Then, $R$ resolves $\Hka$ if and only if $\hbox{ker}(A)\cap\{P=0\}=\{0\}$.
\label{cor:KerCapP=0}
\end{corollary}

Our focus in what remains of this section is to better characterize the non-trivial roots of the polynomial system $\{P=0\}$. To do so, we rely on the concepts of polynomial ideals and (reduced) Gr\"obner bases, and the following fundamental result from algebraic geometry. For a primer to these and other concepts on which our results rely see~\cref{app:1}.

\begin{theorem} (Hilbert's Weak Nullstellensatz~\cite[\S4.1]{Cox_Little_OShea:2015}.)
\label{thm:weak_null}
For any non-empty finite set of polynials $P$, $\{P=0\} = \emptyset$ if and only if $\{1\}$ is the reduced Gr{\"o}bner basis of $I(P)$, the ideal generated by $P$.
\end{theorem}

Define for each $i=1,\ldots,k$:
\begin{eqnarray}
\label{def:Bi}
B_i &:=& \Big\{z_j^3-z_j,\hbox{ for }j=(i-1)a+1,\ldots,ia\Big\}\\
&&\qquad\bigcup\left\{\sum_{j=(i-1)a+1}^{ia}z_j,\Big(2-\sum_{j=(i-1)a+1}^{ia}z_j^2\Big)\cdot\sum_{j=(i-1)a+1}^{ia}z_j^2\right\}.
\notag
\end{eqnarray}
Observe that $B_i$ is a set of polynomials in $(z_{(i-1)a+1},\ldots,z_{ia})$, i.e. the $i$-th subvector of $z$; in particular, each of these polynomials may be regarded as a function of $z$, and $B_1,\ldots,B_k$ partition $P$, i.e. $P=\sqcup_{i=1}^kB_i$. Accordingly, we call $B_i$ the $i$-th \underline{b}lock of $P$, and denote the reduced Gr\"obner basis of $B_i$ as $G_i$. The computational advantage of these observations is revealed by the following results.

\begin{lemma}
$G=\cup_{i=1}^kG_i$ is the reduced Gr\"obner bases of $P$ in equation~\cref{def:P}. Furthermore, $G_i$ may be obtained from $G_1$ using the change of variables:
\begin{equation}
(z_1,\ldots,z_a)\longrightarrow(z_{(i-1)a+1},\ldots,z_{ia}).
\label{ide:varchange}
\end{equation}
\label{lem:groeb_block}
\end{lemma}
\begin{proof}
The case with $k=2$ follows from~\cite[Proposition 2]{Cox_Little_OShea:2015} due to the fact that no variable and hence no polynomial is shared between the blocks of $P$. A straightforward inductive argument in $k\ge2$ then shows that $\cup_{i=1}^kG_i$ is the reduced Gr\"obner bases of $P$. Finally, note that $B_1$ is up to the change of variables in equation~(\ref{ide:varchange}) identical to $B_i$; in particular, since Buchberger's algorithm (\cref{algo:1}) and the Gr\"obner basis reduction algorithm (\cref{algo:2}) build upon polynomial division, the reduced Gr\"obner bases of $B_i$ may be obtained from that of $B_1$ using the same change of variables.
\end{proof}

\begin{lemma}
The reduced Gr\"obner bases of $B_1$ under the lexicographic ordering is
\begin{equation}
G_1=\left\{\sum\limits_{i=1}^a z_i\right\}\bigcup_{2\le i\le a}\{z_i^3-z_i\}\bigcup_{2\le i<j\le a}\{z_i^2z_j+z_iz_j^2\}\bigcup_{2\le i<j<\ell\le a}\{z_iz_jz_\ell\}.
\end{equation}
\label{lem:G1explicit}
\end{lemma}

\begin{proof}
Let $G$ be the set of polynomials on the right-hand side above. Since $G$ depends on $a$ but not on the parameter $k$ of $\Hka$, and the identity for $a\in\{2,3,4\}$ can be checked using algorithms~\ref{algo:1} and~\ref{algo:2}, without loss of generality we assume in what follows that $a>5$.

Since reduced Gr\"obner basis are unique, it suffices to show that (i) $I(G)=I(B_1)$; and that for all $f,g\in G$: (ii) the reduction of $\hbox{Spoly}(f,g)$ by $G$ is 0; (iii) $LC(f)=1$; and (iv) if $f\in G\setminus\{g\}$ then no monomial of $f$ is divisible by $LM(g)$. We omit the tedious but otherwise straightforward verification of properties (ii) and (iv). Since property (iii) is trivially satisfied, it only remains to verify property (i).

To prove $I(G) = I(B_1)$, it suffices to show that $\{G=0\}=\{B_1=0\}$. Indeed, the polynomials of the form $z_iz_jz_\ell$ imply that if $z\in\{G=0\}$ then $(z_2,\ldots,z_{a-1})$ has at most two non-zero coordinates. In the case two of these coordinates are non-zero, say $z_i$ and $z_j$, the polynomials $z_i^3-z_i=z_i(z_i-1)(z_i+1)$, $z_j^3-z_j=z_j(z_j-1)(z_j+1)$, and $z_i^2z_j+z_iz_j^2=z_iz_j(z_i+z_j)$ imply that $(z_i,z_j)=(1,-1)$ or $(z_i,z_j)=(-1,1)$; in particular, because we must have $\sum_{\ell=1}^az_\ell=0$, $z_1=0$. Instead, if exactly one of the coordinates in $(z_2,\ldots,z_{a-1})$ is non-zero, say  $z_j$, then the polynomial $\sum_{\ell=1}^az_\ell$ together with $z_j^3-z_j$ imply that $(z_1,z_j)=(1,-1)$ or $(z_1,z_j)=(-1,1)$. Finally, if $(z_2,\ldots,z_{a-1})=0$ then the polynomial $\sum_{\ell=1}^az_\ell$ implies that $z_1=0$. In all of these three exhaustive cases, it follows that $(z_1,\ldots,z_a)$ is identically zero, or it has exactly one 1 and one (-1) coordinate and all other coordinates vanish; in other words, $(z_1,\ldots,z_a)$ is a difference of two canonical vectors in $\mathbb{R}^a$. Since this is precisely the constraint imposed on this subvector of $z$ by the polynomials in $B_1$, we obtain that $\{G=0\}=\{B_1=0\}$ i.e. $I(G)=I(B_1)$.
\end{proof}

A minor issue for using the Weak Nullstellensatz in our setting is that the polynomials in $P$ have no constant terms; in particular, $0\in\{P=0\}$. To exclude this trivial root, observe that if $z\in P$ then $\sum_{j=(i-1)a+1}^{ia}z_j^2\in\{0,2\}$, for each $i=1,\ldots,k$. As a result, if $z$ is a non-trivial root of $\{P=0\}$ then $\sum_{j=1}^{ka}z_j^2=2i$ for some $i$. This motivates to introduce the auxiliary polynomial:
\begin{equation}
f(z):=\Big(\sum_{j=1}^{ka}z_j^2\Big),
\label{def:f(z)}
\end{equation}
so that $R$ resolves $\Hka$ if and only if $\hbox{ker}(A)\cap\{P=0\}\cap\{f-2i=0\}=\emptyset$ for all $i=1,\ldots,k$.

\begin{lemma}
Consider a (finite) reduced Gr\"obner basis $G \neq \{1\}$ and a polynomial $f$. If $f \xrightarrow{G} r$ then, for each $c\in\mathbb{R}$, $(f+c) \xrightarrow{G} (r+c)$.
\end{lemma}

\begin{proof}
Let $G=\{g_1,\ldots,g_n\}$. Without loss of generality fix a constant $c\ne0$. Note that $G$ contains no constant polynomial (except for $0$) because $G \neq \{1\}$ hence $1\notin G$. As a result, the leading monomial of each $g_i$ does not divide $c$, hence $c\xrightarrow{G}c$. Since $f \xrightarrow{G} r$, and reductions by a Gr\"obner basis are unique, $(f+c) \xrightarrow{G} (r+c)$ as claimed. 
\end{proof}

The following is now a direct consequence of the lemma.

\begin{corollary}
Let $G$ be the reduced Gr\"obner basis of $P$ in equation~(\ref{def:P}). If $f$ is as defined in~\cref{def:f(z)} and $f \xrightarrow{G} r$ then, for each $i = 1,2,\ldots,k$, $(f-2i)\xrightarrow{G} (r-2i)$.
\label{cor:rem}
\end{corollary}

The results from this section allow for a computational method for checking resolvability on $\mathbb{H}_{k,a}$. Lemmas~\ref{lem:groeb_block} and~\ref{lem:G1explicit} are used to construct the reduced Gr\"obner basis $G$ directly, and~\cref{cor:rem} efficiently removes the trivial solution from consideration in the criteria provided by~\cref{thm:weak_null}. Altogether these results significantly reduce the number of polynomial reductions required to assess the resolvability of a set of nodes on $\mathbb{H}_{k,a}$.

\subsection{Illustrative Example (Continuation)} We saw in~\Cref{subsec:Illustrative} that $R_0=\{02,11\}$ does not resolve $H_{2,3}$ whereas $R_1=R_0\cup\{22\}=\{02,11,22\}$ does. We can double-check these assertions using~\cref{cor:KerCapP=0} as follows. 

First, recall that for $H_{2,3}$ the variable $z$ is 6-dimensional and should be decomposed in the form $z=\big((z_1,z_2,z_3),(z_4,z_5,z_6)\big)$. Next, the kernel of the matrix given by the corollary for $R_0$ (denoted as $A_0$, see Eq.~\cref{def:A0}) is described by the linear system:
\[\left\{\begin{array}{rcl}
z_1+z_6 &=& 0;\\
z_2+z_5 &=& 0.
\end{array}\right.\]
On the other hand, the roots in $\{P=0\}$ given by~\cref{cor:KerCapP=0} correspond to the polynomial system: 
\[\left\{\begin{array}{ccl}
0 &=& z_1^3-z_1;\\
0 &=& z_2^3-z_2;\\
0 &=& z_3^3-z_3;\\
0 &=& z_1 + z_2 + z_3;\\
0 &=& (2-z_1^2-z_2^2-z_3^2)\cdot(z_1^2+z_2^2+z_3^2);\\
\hline
0 &=& z_4^3-z_4;\\
0 &=& z_5^3-z_5;\\
0 &=& z_6^3-z_6;\\
0 &=& z_4 + z_5 + z_6;\\
0 &=& (2-z_4^2-z_5^2-z_6^2)\cdot(z_4^2+z_5^2+z_6^2);
\end{array}\right.\]
where the horizontal line distinguishes between the first and second block of $P$ (see Eq.~(\ref{def:Bi})). Finally, recall the auxiliary polynomial given by equation~\cref{def:f(z)}:
\[f(z)=z_1^2+z_2^2+z_3^2+z_4^2+z_5^2+z_6^2.\]
Assuming the lexicographic order over the monomials, one can determine that the reduced Gr\"obner basis of $\{A_0z\}\cup P\cup\{f-2\}$ is $\{1\}$; in particular, $\hbox{ker}(A_0)\cap\{P=0\}\cap\{f=2\}=\emptyset$. On the other hand, because the reduced Gr\"obner basis of $\{A_0z\}\cup P\cup\{f-2\}$ is $\{z_1+z_6, z_2+z_5, z_3-z_5-z_6, z_4+z_5+z_6, z_5^2+z_5z_6+z_6^2-1, z_6^3 - z_6\}$, it follows that $\hbox{ker}(A_0)\cap\{P=0\}\cap\{f=4\}\ne\emptyset$ i.e. $\hbox{ker}(A_0)\cap\{P=0\}$ has a non-trivial solution. Consequently, $R_0$ does not resolve $H_{2,3}$.

To confirm that $R_1=R_0\cup\{22\}$ does resolve $H_{2,3}$, note that we only need to add the equation $z_3+z_6=0$ to the previous linear system (the full linear system is now described by the matrix $A_1$, see Eq.~\cref{def:A1}). Using our code, we find that $\hbox{ker}(A_1)\cap\{P=0\}\cap\{f-2\}=\emptyset$ and also that $\hbox{ker}(A_1)\cap\{P=0\}\cap\{f-4\}=\emptyset$ because the associated reduced Gr\"obner basis are both equal to $\{1\}$. As a result, $\hbox{ker}(A_1)\cap\{P=0\}$ has no non-trivial solution, i.e. $R_1$ resolves $H_{2,3}$.

\section{Novel Integer Linear Programming Formulation}
\label{sec:ILP}

For some background about Integer Linear Programming (ILP), see~\cref{app:2}.

In contrast to the ILP approaches of \cite{chartrand2000resolvability,currie2001metric}, our ILP formulation checks the resolvability of a given set rather than searching for minimal resolving sets. Furthermore, it does not pre-compute the distance matrix of a Hamming graph. As before, fix $\Hka$ and a subset of vertices $R$. Letting $z=(z_1,\ldots,z_{ka})$ and using the polynomial set $P$ from equation~(\ref{def:P}), we leverage~\cref{lem:polsys} (with $A$ as in equation~(\ref{def:A}), each row corresponding to a vertex in $R$) to reformulate~\cref{thm:Az=0} as follows:
\begin{equation}
R \text{ does \underline{not} resolve } \Hka \quad \iff \quad \exists z \neq 0 \;\text{such that}\; z\in\hbox{ker}(A)\cap\{P=0\}.
\label{eq:resolveIff}
\end{equation}

To formulate this as an ILP, we use the following result.
\begin{lemma}
Define
\[\mathcal{I}:= \bigcap_{i=1}^k\left\{z\in\mathbb{Z}^{ak}\hbox{ such that }\sum\limits_{j=(i-1)a+1}^{ia} z_j = 0\hbox{ and }  \sum\limits_{j=(i-1)a+1}^{ia} |z_j| \le 2\right\}.\]
Then $\mathcal{I}$ is the intersection of a closed convex polyhedron with the integer lattice $\mathbb{Z}^{ak}$, and  $z\in\{P=0\}$ if and only if $z\in\mathcal{I}$.
\label{lemma:ILP}
\end{lemma}
\begin{proof}
Since the intersection of convex sets is convex, and the intersection of a finite number of polyhedra is a polyhedron, it follows from standard arguments that
\begin{eqnarray*}
\mathcal{J}_1 &:=& \bigcap_{i=1}^k\left\{z\in\mathbb{R}^{ak}\hbox{ such that }\sum\limits_{j=(i-1)a+1}^{ia} z_j = 0\right\};\\
\mathcal{J}_2 &:=& \bigcap_{i=1}^k\left\{z\in\mathbb{R}^{ak}\hbox{ such that } \sum\limits_{j=(i-1)a+1}^{ia} |z_j| \le 2\right\};
\end{eqnarray*}
are convex subsets of $\mathbb{R}^{ak}$, and $\mathcal{J}_1$ is a polyhedron. We claim that $\mathcal{J}_2$ is also a polyhedron, for which it suffices to check that each set in the intersection that defines it is a polyhedron. Without loss of generality, we do so only for the case with $i=1$. Indeed, because $\{z\in\mathbb{R}^{ak}\hbox{ such that } \sum_{j=1}^a |z_j| \le 2\}$ is invariant under arbitrary coordinate sign flips, we have that
\[\left\{z\in\mathbb{R}^{ak}\hbox{ such that } \sum_{j=1}^a |z_j| \le 2\right\}=\bigcap_{w\in\{-1,1\}^{ak}}\left\{z\in\mathbb{R}^{ak}\hbox{ such that } \sum_{j=1}^a w_jz_j \le 2\right\},\]
which implies that $\mathcal{J}_2$ is also a polyhedron. Since $\mathcal{I}=(\mathcal{J}_1\cap\mathcal{J}_2\cap\mathbb{Z}^{ak})$, the first part of the lemma follows.

From the proof of~\cref{lem:polsys} it is immediate that $\{P=0\}\subset\mathcal{I}$. To show the converse inclusion, observe that $\{P=0\}=\cap_{i=1}^k\{B_i=0\}$ where the $B_i$'s are as defined in equation~(\ref{def:Bi}). To complete the proof, it suffices therefore to show that $\mathcal{I}_i\subset\{B_i=0\}$, where
\[\mathcal{I}_i:=\left\{z\in\mathbb{Z}^{ak}\hbox{ such that }\sum\limits_{j=(i-1)a+1}^{ia} z_j = 0\hbox{ and }  \sum\limits_{j=(i-1)a+1}^{ia} |z_j| \le 2\right\}.\]
Indeed, if $z\in\mathcal{I}_1$ then, because the coordinates of $z$ are integers, the condition $\sum_{j=1}^a |z_j| \le 2$ implies that $|z_j|\in\{0,1,2\}$ for $j=1,\ldots,a$. If $|z_j|=2$ for some $j$ then $\sum_{j=1}^az_j=\pm2$, which is not possible. Thus $z_j\in\{0,\pm1\}$ for $j=1,\ldots,a$; in particular, $z_j^3-z_j=0$. On the other hand, the condition $\sum_{j=1}^az_j=0$ implies that the number of 1's and (-1)'s in $(z_1,\ldots,z_a)$ balance out; in particular, since $\sum_{j=1}^a |z_j| \le 2$, either $(z_1,\ldots,z_a)$ vanishes, or it has exactly one 1 and one (-1) entry and all other entries vanish; in particular, $(2-\sum_{j=1}^az_j^2)\cdot\sum_{j=1}^az_j^2=0$. Thus, $z\in\{B_1=0\}$. The case for $i>1$ is of course the same.
\end{proof}

\begin{remark}
With current ILP solvers, one can impose that $z\in\{0,\pm1\}^{ak}$ simply as $|z_i|\le1$ for $i=1,\ldots,ak$. On the other hand, while a constraint like $\sum_{j=1}^a |z_j| \le 2$ is clearly polyhedral, it is not in the form of an affine equality or inequality suitable for ILP solvers. Nevertheless, standard reformulation techniques can convert this into a set of affine equalities and inequalities in a higher dimensional space. For example, in the product space with variables $(\tilde{z},w)$, we can write the constraint as $\sum_{j=1}^a w_j \le 2$ and $|\tilde{z}_j| \le w_j$ (i.e., $\tilde{z}_j \le w_j$ and $-\tilde{z}_j \le w_j$), which leads to an equivalent formulation of the original ILP. One may handle such reformulations automatically using the Matlab package \texttt{CVX}~\cite{cvx}.
\end{remark}

It only remains to encode the fact that we look for a \emph{nonzero} root in $\{P=0\}$, which we do via the ILP in the following theorem:

\begin{theorem}
A subset of vertices $R$ is \underline{not} resolving on $\Hka$ if and only if the solution to the following ILP is less than zero:
\begin{alignat}{2} 
\label{eq:newILP}
&\min_{z\in\mathbb{R}^{ak}} \;  &  &\sum_{j=1}^{ak} 2^j z_j  \notag \\
&\textnormal{subject to} \;  &  & Az=0 \;\textnormal{and}\; z\in\mathcal{I},
\end{alignat}
where $A$ is defined in equation~\cref{def:A}.
\end{theorem}
\begin{proof}
Using equation~\cref{eq:resolveIff} and \cref{lemma:ILP}, it remains to show that the objective function is less than zero if and only if there is a non-zero feasible $z$. Suppose there is not a non-zero feasible $z$. Clearly $z=0$ is feasible, hence it is the only feasible point for the ILP, and the objective value is zero. Now suppose there is some non-zero feasible $z$. Let $j'$ be the largest non-zero coordinate. Then because $\sum_{j=1}^{j'-1} 2^j < 2^{j'}$, and because each entry is bounded $|z_j|\le 1$, the objective value at this $z$ is non-zero. If the objective value is negative, this proves the value of the ILP is negative; if the objective value is positive, then observe that $(-z)$ is also feasible and has a negative objective value, and hence the value of the ILP is negative.
\end{proof}

\begin{remark}
If the solution to the ILP is less than zero and hence $R$ is not a resolving set, then each optimal vector $z$ is the difference of the column-major ordering of the one-hot encodings of two $\kmers$ which are not resolved by $R$; in particular, a vector that resolves these $\kmers$ needs to be added to $R$ to resolve $\Hka$. 
\end{remark}

\subsection{Practical formulations and roundoff error}
\label{sec:ILP_practical}

When $ak$ is small, it is feasible to directly solve the ILP in equation \cref{eq:newILP}. One issue with larger values of $ak$, besides an obvious increase in run-time, is that the values of $2^j$ in the objective function quickly lead to numerical overflow. A simple fix is to replace each coefficient $c_j = 2^j$ with an independently drawn realization of a standard normal random variable $\mathcal{N}(0,1)$.  Since these new coefficients are independent of the feasible set, if the latter is truly larger than $\{0\}$, the probability that the entire feasible set is in the null-space of the linear function $\sum_{j=1}^{ak}c_jz_j$ is zero. Of course, again due to finite machine precision, this otherwise almost surely exact method may only be approximate. Admittedly, when running the ILP with the random coefficients $c_j$'s, finding an undoubtedly negative solution to the ILP would certify that the set $R$ is not resolving. However, if the solution is just slightly negative or vanishes within machine precision, the assessment about $R$ should be taken with a grain of salt. In this case, one should draw a new set of random coefficients and re-run the ILP to reassess the resolvability of $R$.

Another consideration is that the ILP solver wastes time finding a feasible point with the smallest possible objective, when we only care if there is a feasible point with objective smaller than $0$. Thus we could solve the feasiblity problem
\begin{alignat*}{2} \label{eq:feas1}
&\textnormal{Find}  \;  &  &z\in\mathbb{R}^{ak}  \notag \\
&\textnormal{subject to} \;  &  & Az=0 \;\textnormal{and}\; z\in\mathcal{I} \;\textnormal{and}\;  \langle c, z \rangle < 0
\end{alignat*}
where $c_j = 2^j$ or $c_j \sim \mathcal{N}(0,1)$ as discussed above. (Feasibility problems can be encoded in software by minimizing the $0$ function.) Unfortunately this is not an ILP because $\{ c \mid  \langle c, z \rangle <0 \}$ is not a closed set. We can partially ameliorate this by solving
\begin{alignat}{2} \label{eq:feas2}
&\textnormal{Find} \;  &  &z\in\mathbb{R}^{ak}  \notag \\
&\textnormal{subject to} \;  &  & Az=0 \;\textnormal{and}\; z\in\mathcal{I} \;\textnormal{and}\; \langle c, z \rangle \le -\delta
\end{alignat}
where $\delta>0$ is a small number (our code uses $\delta=10^{-3}$). Finding a feasible point $z$ is then proof that the set $R$ does not resolve $\Hka$. If the solver says the above problem is infeasible, it could be that $\delta$ was too large and hence the computation was inconclusive. In this case, one could run the slower program \cref{eq:newILP}.

\section{Computational Complexity Experiments}
\label{sec:complexity_experiments}

The theoretical framework and algorithms proposed in this paper provide a novel way of approaching resolvability on Hamming graphs. To show the computational feasibility and practicality of our methods, we compare the average run-time of both the ILP and Gr\"obner basis algorithms against the brute force approach for checking resolvability. Our experiments use Python 3.7.3 and SymPy version 1.1.1~\cite{SymPy}, and the commercial ILP solver \texttt{gurobi} ver.~7.5.2~\cite{gurobi}.

In the~\cref{tab:k_a_pairs}, we present the average run-time and standard deviation of the algorithms on reference test sets for Hamming graphs of increasing sizes. \cref{fig:runtime} displays the mean run-times as a function of the graph size, and best linear fit for each method. As seen in the table and figure, the brute force approach is faster on only the smallest Hamming graphs (with fewer than $1000$ nodes) whereas the ILP solution is exceptionally fast even as the Hamming graph grows to more than $6000$ nodes. For small problems, the time taken to solve the ILP is likely dominated by the overhead cost of using \texttt{CVX} to recast the ILP into standard form. The run-time results show a promising improvement in computational time over the brute force approach which will only become more pronounced on massive Hamming graphs. Additionally, the brute force approach is infeasible on these larger graphs due to significant memory costs. 

The ILP algorithm is exceptionally quick, beating all other methods for Hamming graphs with more than 1000 nodes, but it cannot guarantee that a set is resolving. The Gr\"obner basis algorithm by contrast is much slower on average but is a deterministic method of showing resolvability. ILP can be used to quickly determine possible resolving sets which are then verified by the Gr\"obner basis algorithm. In this way, the two methods are symbiotic and cover each other's weaknesses. We illustrate this in the next section. 

\begin{table}
\centering 
\tiny
\begin{tabular}{cS[table-format=4]S[table-format=1.2e-1]S[table-format=1.2e-1]S[table-format=1.2e-1]S[table-format=1.2e-1]S[table-format=1.2e-1]S[table-format=1.2e-1]}
\toprule
& & \multicolumn{2}{c}{Brute Force} & \multicolumn{2}{c}{Gr\"obner Basis} & \multicolumn{2}{c}{ILP} \\
\cmidrule(lr){3-4}  \cmidrule(lr){5-6}  \cmidrule(lr){7-8} 
($k,a$) & {$a^k$} & {Mean} & {SD} & {Mean} & {SD} & {Mean} & {SD} \\
\midrule
(2,2) & 4 & 3.88e-05 & 1.51e-06 & 6.79e-03 & 1.06e-03 & 1.28e-01 & 3.53e-03 \\
(2,4) & 16 & 2.47e-04 & 6.83e-05 & 2.25e-02 & 2.59e-03 & 1.16e-01 & 4.84e-03 \\
(3,3) & 27 & 5.02e-4 & 2.45e-4 & 2.83e-2 & 7.92e-3 & 1.21e-01 & 8.12e-3 \\
(5,2) & 32 & 6.61e-04 & 3.29e-04 & 3.14e-02 & 5.27e-03 & 1.28e-01 & 3.91e-03 \\
(3,5) & 125 & 8.98e-03 & 5.38e-03 & 1.12e-01 & 2.91e-02 & 1.37e-01 & 7.02e-03 \\
(5,3) & 243 & 2.78e-2  & 1.96e-2  & 1.22e-1 & 7.88e-2 &   1.20e-1 & 8.12e-3 \\
(8,2) & 256 & 2.85e-02 & 2.21e-02 & 9.87e-02 & 1.96e-02 & 1.17e-01 & 1.59e-03 \\
(4,4) & 256 & 3.13e-02 & 1.97e-02 & 1.27e-01 & 3.90e-02 & 1.37e-01 & 9.58e-03 \\
(5,5) & 3125 & 5.19e+00 & 3.17e+00 & 4.00e+00 & 3.54e+00 & 1.35e-01 & 1.09e-02 \\
(12,2) & 4096 & 6.28e+00 & 5.34e+00 & 2.93e-01 & 7.24e-02 & 1.24e-01 & 2.39e-03 \\
(6,4) & 4096 & 7.78e+00 & 4.65e+00 & 7.73e-01 & 3.67e-01 & 1.52e-01 & 8.99e-03 \\
(8,3) & 6561 & 2.02e+01 & 1.40e+01 & 1.12e+01 & 1.47e+01 & 1.62e-01 & 1.41e-02 \\
\bottomrule
\end{tabular}
\caption{Time in seconds required to determine resolvability for each technique. Fifty resolving and fifty non-resolving sets, selected uniformly at random, were considered for each Hamming graph $\Hka$. Means and standard deviations consider five replicates per set.}
\label{tab:k_a_pairs}
\end{table}

\begin{figure}
    \centering
    \includegraphics[width=2.7in]{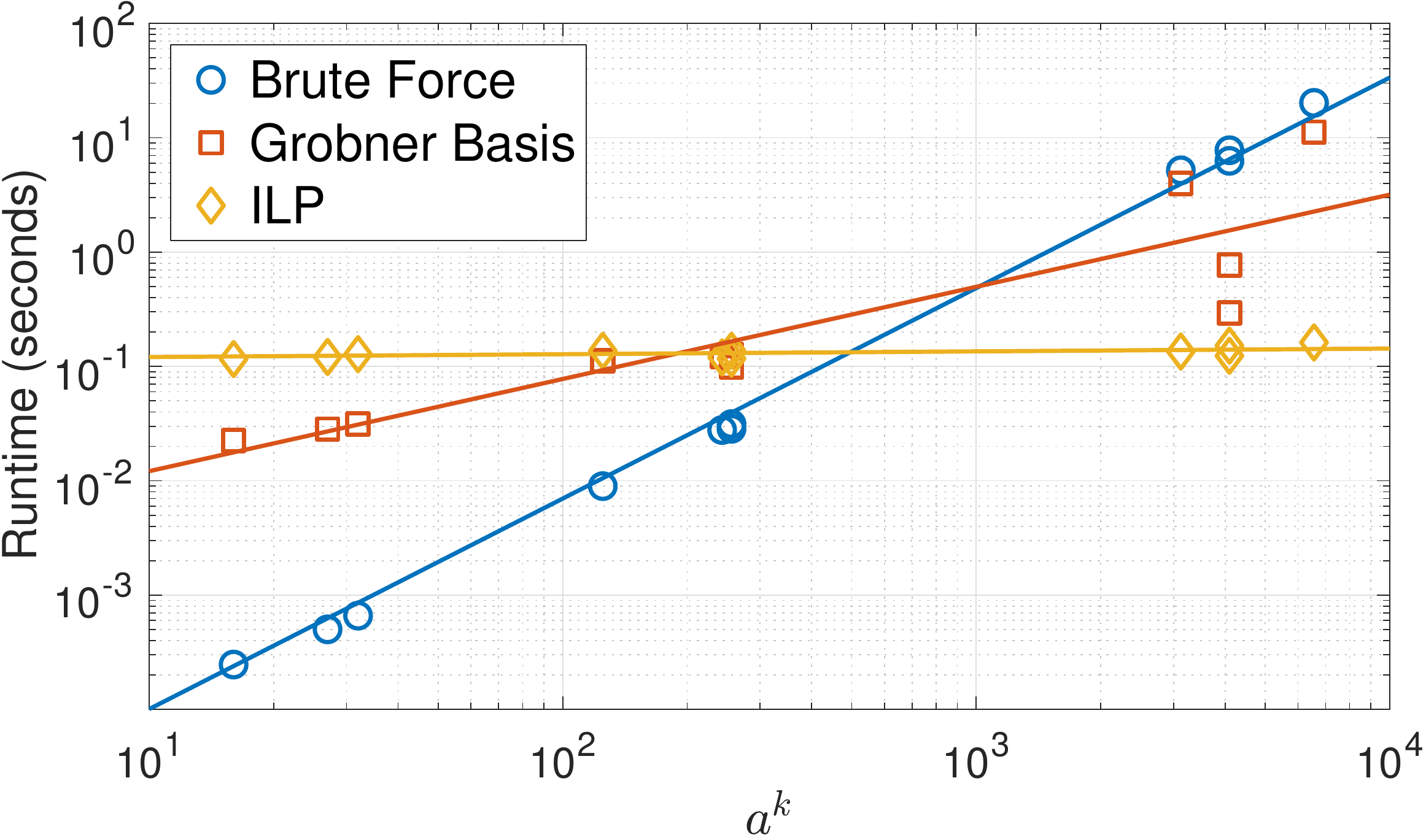}
    \caption{Data from~\cref{tab:k_a_pairs} with lines of best fit (on log-transformed data) for each method.}
    \label{fig:runtime}
\end{figure}

\section{Low-dimensional Protein Representations}
\label{sec:protein_representation}

Symbolic information pervades modern data science. With the advent and popularization of high-throughput sequencing assays, this is particularly true in the field of computational biology where large volumes of biological sequence data have become critical for studying and understanding the behavior of cells. Analysis of these sequences, however, presents significant challenges. One major issue is that many powerful analysis techniques deal with numeric vectors, not arbitrary symbols. As a result, biological sequence data is typically mapped to a real space before such methods are applied. Two of the most common mappings use K-mer count~\cite{leslie2002spectrum} and one-hot encodings (also called binary vectors)~\cite{cai2003support}. K-mer count vectors represent symbolic sequences by their counts of each possible K-mer.

Resolving sets can be used to define low-dimensional mappings as well. To fix ideas we focus on octapeptides, that is proteins composed of 8 amino acids. With a total of 20 possible amino acids (which we represent as {\ttfamily {\footnotesize a,r,n,d,c,q,e,g,h,i,l,k,m,f, p,s,t,w,y,v}}) and imposing the Hamming distance across these sequences, we have the Hamming graph $\mathbb{H}_{8,20}$. This graph is massive. It has $25.6$ billion vertices and roughly $1.9$ trillion edges rendering most methods of discovering small resolving sets, including the ICH algorithm, useless. Utilizing a constructive algorithm, a resolving set of size 82, which we call $R$, was discovered for $\mathbb{H}_{8,20}$ in~\cite{TilLla19}. However, it is not known whether $R$ contains a proper subset that is still resolving. Here, we address this problem applying the results of sections~\ref{sec:grobner} and \ref{sec:ILP}.

Starting with lower and upper bounds $L=1$ and $U=82$ respectively, we implement a binary search for $\beta(\mathbb{H}_{8,20})$. With $s=\frac{L+U}{2}$ as the current subset size to check, up to 1000 subsets of $R$ are selected at random. The ILP approach (\cref{sec:ILP}) then provides an efficient method for testing the feasibility problem outlined in \cref{thm:Az=0} for these subsets. If any subset passes this test, the upper bound is set to $s$. Otherwise, $s$ becomes the lower bound. This process is repeated until $L=(U-1)$. Following this procedure, we found the following set of size $77$:
\[
r:=\left\{
{
\scriptsize
\begin{tabular}{lllllll} 
aaaraaaa, & arwaaaaa, & ccchhhhh, & ccchhhhi, & ccchhhia, & ccchhiaa, & ccchiaaa,\\
ccciaaaa, & cnsaaaaa, & dddeeeee, & dddeeeeg, & dddeeega, & dddeegaa, & dddegaaa,\\
dddgaaaa, & dhfaaaaa, & eagaaaaa, & eeefaaaa, & eeemfaaa, & eeemmfaa, & eeemmmfa,\\
eeemmmmf, & eeemmmmm, & fffaaaaa, & gggppppp, & gggpppps, & gggpppsa, & gggppsaa,\\
gggpsaaa, & gggsaaaa, & hhhttttt, & hhhttttw, & hhhtttwa, & hhhttwaa, & hhhtwaaa,\\
hhhwaaaa, & hpvaaaaa, & iiivaaaa, & iiiyvaaa, & iiiyyvaa, & iiiyyyva, & iiiyyyyv,\\
iiiyyyyy, & kkkaaaaa, & klqaaaaa, & lllaaaaa, & mkyaaaaa, & mmmaaaaa, & nnnccccc,\\
nnnccccq, & nnncccqa, & nnnccqaa, & nnncqaaa, & nnnqaaaa, & nstaaaaa, & pppaaaaa,\\
qpkaaaaa, & qqqkaaaa, & qqqlkaaa, & qqqllkaa, & qqqlllka, & qqqllllk, & qqqlllll,\\
qyeaaaaa, & rrrdaaaa, & rrrndaaa, & rrrnndaa, & rrrnnnda, & rrrnnnnd, & rrrnnnnn,\\
sisaaaaa, & svtaaaaa, & ttcaaaaa, & vfraaaaa, & wmpaaaaa, & wwdaaaaa, & yglaaaaa
\end{tabular}
}
\right\}.
\]
Since the ILP formulation does not guarantee that this set is resolving, we verified the result using a parallelized version of the Polynomial Roots Formulation (\cref{sec:grobner}) so that the Gr\"obner bases of multiple auxiliary polynomials (Eq.~\cref{def:f(z)}) may be determined simultaneously. Thus, we have found a set $r\subset R$ of size 77 that resolves $\mathbb{H}_{8,20}$; in particular, $\beta(\mathbb{H}_{8,20})\le77$, which improves the bound of~\cite{TilLla19}, and all $25.6$ billion octapeptides may be uniquely represented with only 77 dimensions. In contrast, a $2$-mer count vector representation would require 400 dimensions and a one-hot encoding 160 dimensions.

\begin{remark}
We replicated the verification of $r$ as a resolving set of $H_{8,20}$ using our Polynomial Roots Formulation 10 times across 32 computer cores. Overall, a maximum of approximately 380 megabytes of memory per core (SD $\sim 0.5$ MB) and 6 hours and 20 minutes (SD $\sim142$ s) were required to demonstrate the resolvability of $r$. Memory usage was determined using the Slurm workload manager \verb|sacct| command and \verb|maxRSS| field, while time was measured using Python's \verb|time| module.
\end{remark}

\newpage

\appendix

\section{Gr\"obner Basis}
\label{app:1}

In what follows, $z=(z_1,\ldots,z_d)$ is a $d$-dimensional variable and, unless otherwise stated, polynomials are functions of $z$ with real coefficients.

\subsection{Polynomial Ideals} 
\label{app:ideals}

A \textit{polynomial ideal} $I$ is a non-empty set of polynomials with the property that if $f,g\in I$ and $c\in\mathbb{R}$ then $cf, f+g, fg\in I$. 

The polynomial ideal associated with a non-empty and finite set $P = \{p_1,\ldots,p_m\}$ of polynomials in the variable $z$ is the set defined as
\[I(P):=\left\{\sum_{i=1}^mq_i\cdot p_i,\hbox{ with $q_1,\ldots,q_m$ polynomials in $z$}\right\}.\]

Polynomial ideals are useful to characterize the complex numbers $z$ such that $p_i(z)=0$ for all $i=1,\ldots,m$. Indeed, $z\in\{P=0\}$ if and only if $z\in\{I(P)=0\}$.

\subsection{Monomial Orderings} 
\label{app:order}

A \textit{monomial} (in the variable $z$) is any product of the form $z_1^{a_1}\cdots z_d^{a_d}$, where $a_1,\ldots,a_d\ge0$ are integers. This product is often written $z^a$ with  $a=(a_1,\ldots,a_d)$.

A \textit{monomial ordering} is a total ordering of the monomials such that (i) if $z^a<z^b$ then, for any monomial $z^c$, $z^{a+c}<z^{b+c}$; and (ii) $1<z^a$ when $z^a\ne1$. 

A common example of monomial ordering is the so-called \textit{lexicographic order}. Under this ordering, $z^a<z^b$ if there is an index $i$ such that $a_j=b_j$ for all $1\le j<i$ but $a_i<b_i$. Another example is the \textit{graded lexicographic order} under which $z^a<z^b$ if either (i) $\sum_{i=1}^da_i<\sum_{i=1}^db_i$; or (ii) $\sum_{i=1}^da_i=\sum_{i=1}^db_i$ but $z^a$ is smaller than $z^b$ under the lexicographic order. Both of these orderings can be reversed giving the \textit{reversed lexicographic order} and the \textit{graded reverse lexicographic order}, respectively.

\textbf{In what remains of~\cref{app:1}, a fixed monomial ordering is assumed.} In particular, each non-zero polynomial $p$ may be uniquely written in the form $p=\sum_{i=1}^t c_i m_i$, where $t\ge1$ is an integer, $c_1,\ldots,c_t$ are real coefficients, and $m_1>\cdots>m_t$ are monomials in descending order. This allows us to define $\hbox{LM}(p) := m_1$ (the \textit{leading monomial} of $p$), $\hbox{LC}(p) := c_1$ (the \textit{leading coefficient} of $p$), and $\hbox{LT}(p) := c_1 m_1$ (the \textit{leading term} of $p$).

\subsection{Polynomial Reductions} 
\label{app:polred}

For a given non-empty set $P$ of polynomials, every polynomial $f$ can be represented in the form $f=(r+g)$, with $r,g$ polynomials such that $g\in I(P)$. Representing $f$ in this form is called \textit{reducing} $f$ by $P$. The term $r$ is called the \textit{reduction} of $f$ by $P$; which is expressed in writing as $f\stackrel{P}{\rightarrow} r$. Observe that reductions are typically not unique. This is because, if $f=(r+g)$, with $g\in I(P)$, then $f=(r-h)+(g+h)$ for any polynomial $h$, however, $(g+h)\in I(P)$ when $h\in I(P)$.

Reductions can be computed using multivariate long-division as follows. First, set $r=0$ and $g=f$, and look for the smallest index $i$ such that $LM(p_i)$ divides $LM(g)$. If such an index exists, set $g=g-LT(g)\cdot p_i/LT(p_i)$ so that the $LT(g)$ and the $LT(p_i)$ cancel. Otherwise, set $g=g-LT(g)$ and $r=r+LT(g)$. Continue this process until $g=0$. This will produce a remainder $r$ where no monomial of $r$ is divisible by any $LM(p_i)$.

In general, reductions of the form $f\stackrel{P}{\rightarrow} r$ with $r\notin I(P)$ are also not unique. This is because the polynomial $r$ obtained by the long-division depends on the order in which the polynomials in $P$ are indexed. This lack of uniqueness is the primary motivation for Gr\"obner bases.

\subsection{Buchberger's Criterion} 
\label{app:Buch}

In this section we give a characterization of Gr\"obner bases due to Buchberger~\cite{Cox_Little_OShea:1998}.

The \textit{least common multiple} between $z^a$ and $z^b$ is the monomial $\hbox{LCM}(z^a,z^b):=z^c$ with $c=(\max\{a_1,b_1\},\ldots,\max\{a_d,b_d\})$. 

Given two polynomials $p$ and $q$ such that $\hbox{LCM}(\hbox{LM}(p),\hbox{LM}(q)) = z^c$, their \textit{S-polynomial} is defined as
\[\hbox{Spoly}(p,q) := z^c\cdot\left(\frac{p}{\hbox{LT}(p)} - \frac{q}{\hbox{LT}(q)}\right).\]

A non-empty set $G = \{g_1,\ldots,g_n\}$ of polynomials is called a \textit{Gr\"obner basis} for a polynomial ideal $I$ if (i) $I(G)=I$, and (ii) $\hbox{Spoly}(g_i,g_j)\stackrel{G}{\rightarrow} 0$ for all $g_i,g_j\in G$. 

Gr\"obner bases have the following property with regards to reductions: $f\stackrel{G}{\rightarrow} 0$ if and only if $f \in I(G)$, otherwise $f\stackrel{G}{\rightarrow}r$ for some $r \notin I(G)$.

A Gr\"obner basis $G$ is called \textit{reduced} if for all $g_i\in G$, (iii) $LC(g_i)=1$, and (iv) for all $g_j\in G\setminus\{g_i\}$, no monomial of $g_j$ is divisible by $LM(g_i)$.

Unlike Gr\"obner bases, the reduced Gr\"obner basis of a polynomial ideal is unique.

\subsection{Buchberger's Algorithm}
\label{app:BuchAlg}

This is a method for generating a Gr\"obner basis for a polynomial ideal $I(P)$ based on Buchberger's criterion (see~\cref{algo:1}). The key idea of the algorithm is to add to the initially empty Gr\"obner basis S-polynomials of pairs in $P$ which do not reduce to $0$. This by construction satisfies Buchberger's criterion and hence produces a Gr\"obner basis.

Buchberger's Algorithm, however, does not necessarily produce a reduced Gr\"obner basis. Such a reduction can be computed using~\cref{algo:2}. This algorithm is the simplest but also least efficient for computing Gr\"obner bases.

The computation of Gr\"obner bases is an active field of research with many different approaches. There are matrix reduction based algorithms, such as Faug\'ere's F4 algorithm~\cite{Faugere:1999}, as well as signature based algorithms, like Faug\'ere's F5 algorithm and its variants F5C and F5B~\cite{Faugere:2002,Eder:2010,Sun_Wang:2018}.

\begin{algorithm}
\begin{small}
\caption{Buchberger's algorithm for computing a Gr\"obner basis}
    \begin{algorithmic}
        \STATE{Input: $P$, $>$}
            \STATE{$G$ = $P$}
            \STATE{$SP = \{\hbox{Spoly}(g_i,g_j) \vert \forall i < j$,  $g_i,g_j \in G\}$}
            \WHILE{$SP$ not empty}
                \STATE Select $S \in SP$
                \STATE $SP = SP \setminus \{S\}$
                \STATE $r = S\stackrel{G}{\rightarrow}r$
                \IF{$r \neq 0$}
                    \STATE Add $\hbox{Spoly}(r,g_i)$ to $SP$ for each $g_i \in G$ 
                    \STATE $G$ = $G \cup \{r\}$
                \ENDIF
            \ENDWHILE
            \RETURN G
    \end{algorithmic}
\label{algo:1}
\end{small}
\end{algorithm}

\begin{algorithm}
\begin{small}
\caption{Algorithm for reducing a Gr\"obner basis}
    \begin{algorithmic}
        \STATE{Input: $G$, $>$}
            \STATE For each $g_i \in G$, $g_i=\frac{g_i}{LC(g_i)}$ 
            \FOR{$g_i \in G$}
                \STATE $H=G\setminus\{g_i\}$
                \STATE $r=g_{i}\stackrel{H}{\rightarrow}r$
                \IF{$r \neq 0$}
                    \STATE $g_i = r$
                \ELSE
                    \STATE $G=G\setminus\{g_i\}$
                \ENDIF
            \ENDFOR
            \RETURN G
    \end{algorithmic}
\label{algo:2}
\end{small}
\end{algorithm}

\section{Linear Programming}
\label{app:2}

A linear program (LP) is any optimization problem that minimizes or maximizes a linear objective function over a (possibly unbounded) closed convex polyhedron in Euclidean space, i.e., a finite number of affine equalities and inequalities (but not strict inequalities). In the combinatorial model of complexity, LP's are known to be solvable in polynomial time~\cite{khachiyan1979polynomial}. If the variables are constrained to be integers, the LP is an integer linear program (ILP), or more generally, if some of the variables are constrained to be integers and others are not, it is a mixed integer linear program (MILP); one usually does not make the distinction between ILP and MILP, since all ILP solvers also solve MILP.

Unlike LPs, ILP can encode NP-Hard decision problems, and thus in the worst case they are intractable for large problems. However, because ILPs usually admit relaxation (e.g. dropping some constrains), one can find upper and lower bounds on the value of an ILP, and thus, when the bounds meet, produce a certificate of optimality. Standard techniques such as branch-and-bound produce a large tree of different relaxations of the ILP, each node a LP and thus solvable. Once all leaf nodes are visited, one has a guarantee of the optimal solution, but there may be combinatorially many leaf nodes. Modern ILP solvers use clever heuristics to determine the order in which to traverse the tree, and with a good choice, entire branches of the tree can quickly be pruned. Thus ILP solvers combine the speed of heuristics with provable optimality. Furthermore, because ILPs are common in industry, there are many high-quality commercial solvers with excellent implementations such as \texttt{gurobi}~\cite{gurobi}.

Because in the worst-case ILPs can take exponentially long to solve, they are sometimes overlooked in the mathematician's toolbox. Nevertheless, because they are general, they are the focus of much research and thus the software to solve them continues to improve rapidly. To quote Dimitris Bertsekas, 
``in the last twenty-five years (1990--2014), algorithmic advances in integer optimization combined with hardware improvements have resulted in an astonishing 200 billion factor speedup in solving MIO [mixed-integer optimization] problems'' \cite{Bertsimas13lecture}. Thus the main challenge is recognizing when a problem with a combinatorial flavor can be recast as an ILP.

\section*{Acknowledgments}
This research was partially funded by NSF grant No. 1836914. The authors acknowledge the BioFrontiers Computing Core at the University of Colorado Boulder for providing High-Performance Computing resources (fund\-ed by NIH grant No. 1S10OD012300), supported by BioFrontiers IT group.

\bibliographystyle{siamplain}
\bibliography{references}
\end{document}